\documentclass{theoretics}
\title{
  Improved Hardness Results for Learning Intersections of Halfspaces
}
\ThCSauthor[aff1]{Stefan Tiegel}{stefan.tiegel@inf.ethz.ch}[0009-0001-9719-0036]
\ThCSaffil[aff1]{Massachusetts Institute of Technology}
\ThCSthanks{This article was invited from COLT 2024.~\cite{DBLP:conf/colt/Tiegel24} \\
This project has received funding from the European Research Council (ERC) under the European Union's Horizon 2020 research and innovation programme (grant agreement No 815464). 
This work was completed while the author was a PhD student at ETH Z\"urich.}
\ThCSshortnames{S.~Tiegel}
\ThCSshorttitle{Improved Hardness Results for Learning Intersections of Halfspaces}
\ThCSyear{2026}
\ThCSarticlenum{8}
\ThCSreceived{Mar 19, 2025}
\ThCSaccepted{Oct 6, 2025}
\ThCSpublished{Apr 28, 2026}
\ThCSdoicreatedtrue
\ThCSkeywords{Intersections of halfspaces, weak learning, cryptographic hardness}

\ThCSnewtheostd{problem}
\ThCSnewtheostd{assumption}
\ThCSnewtheostd{fact*}

\usepackage{thm-restate}
\makeatletter
\newcommand\IfRestateTF{%
  \ifx\label\thmt@gobble@label 
    \expandafter\@firstoftwo
  \else
    \expandafter\@secondoftwo
  \fi
}
\makeatother
\newcommand{\RestateRemark}{\IfRestateTF{{\normalfont\bfseries (Restated) }}{}}

\usepackage{verbatim}
\usepackage[T1]{fontenc}
\usepackage{mathtools}
\usepackage{pgfplots}
\usepackage{textcomp}
\usepackage{bm}

\usepackage[capitalise,noabbrev,nameinlink]{cleveref}
\crefname{problem}{Problem}{Problems}
\crefname{assumption}{Assumption}{Assumptions}
\crefname{fact}{Fact}{Facts}
\crefname{figure}{Figure}{Figures}
\usepackage{tikz}
\pgfplotsset{compat=1.18}
\usetikzlibrary{positioning}
\usepackage{newfloat}
\usepackage{bbm}

\usepackage{boxedminipage}
\newcommand{\paren}[1]{(#1)}
\newcommand{\Paren}[1]{\left(#1\right)}

\newcommand{\abs}[1]{\lvert#1\rvert}
\newcommand{\Abs}[1]{\left\lvert#1\right\rvert}

\newcommand{\card}[1]{\lvert#1\rvert}

\newcommand{\Set}[1]{\left\{#1\right\}}

\newcommand{\norm}[1]{\lVert#1\rVert}
\newcommand{\Norm}[1]{\left\lVert#1\right\rVert}

\newcommand{\snorm}[1]{\norm{#1}^2}

\newcommand{\iprod}[1]{\langle#1\rangle}

\newcommand{\Esymb}{\mathbb{E}}
\newcommand{\Psymb}{\mathbb{P}}

\DeclareMathOperator*{\E}{\Esymb}

\newcommand{\suchthat}{\;\middle\vert\;}

\newcommand\bdot\bullet

\DeclareMathOperator{\Ind}{\mathbf 1}

\DeclareMathOperator{\poly}{poly}

\DeclareMathOperator{\sign}{sign}

\newcommand{\iid}{i.i.d.\xspace}

\newcommand{\Z}{\mathbb Z}
\newcommand{\N}{\mathbb N}
\newcommand{\R}{\mathbb R}

\newcommand{\cA}{\mathcal A}

\newcommand{\cD}{\mathcal D}

\newcommand{\cS}{\mathcal S}

\newcommand{\cU}{\mathcal U}

\newcommand{\cX}{\mathcal X}

\renewcommand{\leq}{\leqslant}

\renewcommand{\geq}{\geqslant}

\let\epsilon=\varepsilon

\newcommand{\sse}{\subseteq}

\newcommand{\e}{\epsilon}

\allowdisplaybreaks
\newcommand*{\Id}{\mathrm{Id}}

\DeclareMathOperator{\Span}{Span}

\newcommand{\err}{\mathrm{err}}

\newcommand{\clwed}[3]{\mathrm{C}_{#1,#2,#3}}
\newcommand{\hclwed}[4]{\mathrm{H}_{#1,#2,#3,#4}}
\newcommand{\nhclwed}[4]{\mathrm{NH}_{#1,#2,#3,#4}}

\newcommand{\clwem}[3]{\mathrm{CLWE}\Paren{#1,#2,#3}}
\newcommand{\hclwem}[4]{\mathrm{HCLWE}\Paren{#1,#2,#3,#4}}
\newcommand{\nhclwem}[4]{\mathrm{NHCLWE}\Paren{#1,#2,#3,#4}}

\newcommand{\negl}{\mathrm{negl}}

\newcommand{\TVD}[2]{\mathrm{TVD}(#1,#2)}

\newcommand{\Be}[1]{\mathrm{Be}(#1)}
\newcommand{\BE}[1]{\mathrm{Be}\Paren{#1}}

\begin{document}
\maketitle
\begin{abstract}

We show strong (and surprisingly simple) lower bounds for weakly learning intersections of halfspaces in the improper setting. 
Strikingly little is known about this problem.
For instance, it is not even known if there is a polynomial-time algorithm for learning the intersection of only two halfspaces.
On the other hand, lower bounds based on well-established assumptions (such as approximating worst-case lattice problems or variants of Feige's 3SAT hypothesis) are only known (or are implied by existing results) for the intersection of super-logarithmically many halfspaces~\cite{klivans2009cryptographic,klivans2006cryptographic,daniely2016complexity}.
Efficient learning of intersections of fewer halfspaces has only been ruled out under less standard assumptions~\cite{daniely2021local} (such as the existence of local pseudo-random generators with large stretch).
We significantly narrow this gap by showing that even learning the intersection of $\omega(\log \log N)$ halfspaces in dimension $N$ takes super-polynomial time under standard assumptions on worst-case lattice problems (namely that SVP and SIVP are hard to approximate within polynomial factors).
Further, we give unconditional hardness results in the statistical query framework.
Specifically, we show that for any $k$ (even constant), learning the intersection of $k$ halfspaces in dimension $N$ requires accuracy $N^{-\Omega(k)}$, or exponentially many queries -- in particular ruling out SQ algorithms with polynomial accuracy for intersections of $\omega(1)$ halfspaces.
To the best of our knowledge this is the first unconditional hardness result for learning a super-constant number of halfspaces.

Our lower bounds are obtained in a unified way via a novel connection we make between intersections of halfspaces and the so-called parallel pancakes distribution~\cite{DKS17,bubeck2019adversarial,CLWE} that has been at the heart of many lower bound constructions in (robust)  high-dimensional statistics in the past few years.

\end{abstract}

\section{Introduction}

This work studies the computational complexity of weakly learning intersections of halfspaces in the PAC model~\cite{DBLP:journals/cacm/Valiant84}.
A halfspace $h_w \colon \R^N \rightarrow \Set{\pm1}$, or linear threshold function (denoted by LTF for short), is a function $x \mapsto \sign(\iprod{x,w})$ for some unit vector $w \in \R^N$.
A fundamental question is to what extent we can predict the output of $h_w$ on a fresh example, when given random example-label-pairs $(x,h_w(x))$ (where $x$ can follow an arbitrary distribution).
This problem is very well-understood and known to be solvable in time polynomial in the dimension and the inverse of the desired accuracy~\cite{maass1994fast}.
On the other hand, surprisingly little is known when considering only slightly more complex functions such as a function of a small number of halfspaces.
This holds true, even if the functions are simple functions, such as the $\mathrm{AND}$ function.
Note that the $\mathrm{AND}$ function of several halfspaces corresponds to their intersection since for $(x,y)$ it holds that $y = 1$ if and only if $x$ is classified as positive by all halfspaces.

This class is particularly appealing since, depending on the number of halfspaces, it interpolates naturally between very simple (a single halfspace) and very complex boolean functions (such as polytopes with many facets).
Studying the performance of efficient algorithms in this setting, parametrized by the number of halfspaces, can thus serve as a benchmark of how complex functions we could hope to learn.
Formally, the problem is defined as follows:
\begin{definition}
    Let $k,N \in \N$.
    A distribution $D$ over $\R^N \times \Set{\pm 1}$ is an \emph{intersection of $k$ halfspaces}, if it can be described as follows:
    Let $w_1,\ldots, w_k \in \R^N$ be unit vectors and for $i \in [k]$ let $h_{w_i} = \sign(\iprod{w_i,x})$.
    Let $D_x$ be an arbitrary distribution over $\R^N$.
    A sample $(x,y)$ from $D$ is produced by first drawing $x \sim D_x$ and then setting $y = 1$ if and only if $h_{w_i}(x) = 1$ for all $i$.
\end{definition}

We measure the performance of an algorithm as follows:
For any function $f \colon \R^N \rightarrow \Set{\pm 1}$, we define the misclassification error with respect to a distribution $D$ over $\R^N \times \Set{\pm 1}$ as $\err_D(f) \coloneqq \Psymb_{(x,y) \sim D} (f(x) \neq y)$.
We say an algorithm \emph{weakly} learns $D$, if given \iid samples from $D$, it outputs a function $\hat{f}$ such that $\err_D(\hat{f}) \leq \tfrac 1 2 - \tfrac 1 {\poly(N)}$, for some polynomial.
Intuitively, this means the algorithm does slightly better than randomly guessing the label $y$.
This paper studies to what extent we can hope to weakly learn the intersection of few (with respect to the dimension) halfspaces.

We remark that we do not restrict our algorithm to output an intersection of $k$ (or more) halfspaces but allow that it returns an arbitrary boolean function.
This setting is called \emph{improper learning}.
In contrast, the setting in which the hypothesis needs to be of the same (or a slightly larger) family, is referred to as (semi-)proper learning.
Proving lower bounds against improper learners has proven to be significantly more difficult than against proper learners.
In particular, while it is known how to show $\mathrm{NP}$-hardness (under randomized reductions) of properly learning many natural classes of functions~\cite{feldman2006optimal,feldman2006new,DBLP:conf/focs/GuruswamiR06, MR2644358-Gopalan10}, there are inherent barriers for showing such reductions in the improper setting~\cite{applebaum2008basing}.
In fact, improper learners are known to be strictly more powerful.
For instance, there are concept classes for which it is known that it is $\mathrm{NP}$-hard to find a proper learner but efficient improper learners exist~\cite{DBLP:journals/cacm/Valiant84,DBLP:journals/jacm/PittV88}.\footnote{3-Term DNFs are an example of such a class, since they are known to be efficiently learnable via 3-CNFs}

\paragraph{Previous hardness results}
Indeed, in the proper learning setting it is known that it is $\mathrm{NP}$-hard to learn the intersection of two halfspaces, even if the learner is allowed to output a function that is an intersection of any constant number of halfspaces~\cite{DBLP:conf/focs/AlekhnovichBFKP04}.
In contrast, in the improper setting, despite extensive work on this topic~\cite{klivans2004learning,DBLP:conf/focs/KlivansOS08,klivans2009baum,DBLP:conf/focs/Vempala10a,sherstov2010optimal,sherstov2021hardest}, it is not even known whether there are polynomial-time algorithms for (improperly) learning the intersection of two halfspaces unless we make additional assumption about the marginal distribution $D_x$.
Nor is there any evidence of hardness.\footnote{The only exception are several structural observations~\cite{sherstov2010optimal,sherstov2021hardest}. We will come back to this later.}

Due to the dearth of algorithmic results, researchers have started to look for evidence of hardness.
Most of these are reduction-based, while a few are unconditional but restricted to the statistical query (SQ) model.
The first hardness result is the seminal work of~\cite{klivans2009cryptographic} showing that for any $\e > 0$, weakly learning the intersection of $N^\e$ halfspaces is not possible in polynomial time, assuming hardness of certain worst-case lattice problems that form the basis of a large branch of cryptography (specifically, approximating SVP and SIVP up to polynomial factors, see~\cref{prob:gap_svp,prob:sivp} and the end of this section for precise definitions and a discussion, we also refer to~\cite{peikert2016decade}).\footnote{This was later strengthened to $\log^C (N)$ for some  constant $C > 2$~\cite{klivans2006cryptographic}. See~\cref{sec:relation_prev_work} for a more detailed discussion.}
This was slightly strengthened in~\cite{daniely2016complexity} to showing that learning intersections of $\omega(\log N)$ halfspaces is hard assuming a widely believed variant of Feige's hypothesis about refuting random 3$\mathrm{SAT}$ instances~\cite{MR2121179-Feige02}.

Going beyond this, researchers had to resort to less standard assumptions to show reduction-based hardness for intersections of even fewer halfspaces.
In particular,~\cite{daniely2021local} showed that assuming the existence of so-called \emph{local} pseudo-random generators with polynomial stretch, even learning intersections of $\omega(1)$ halfspaces is hard -- assuming that a specific candidate function actually satisfies these properties, they are able to show that learning $k$ halfspaces takes time at least $n^{\Omega(k)}$.
While this indeed gives some evidence of hardness, we believe verifying these predictions based on more standard assumptions or via unconditional lower bounds in restricted models of computation is an important line of work.
Yet, proving such strong, or even fine-grained results, under more standard assumptions, such as approximating worst-case lattice problems or (variants of) Feige's hypothesis, has remained elusive.
In our work, we make significant progress in this direction, by showing that even learning the intersection of $\omega(\log \log N)$ halfspaces is hard under standard assumptions about approximating SVP and SIVP similar to~\cite{klivans2009cryptographic}.

In terms of unconditional lower bounds,~\cite{klivans2007unconditional} showed that (roughly speaking), restricted to the SQ model, learning intersections of $k$ halfspaces takes time at least $N^{\Omega(k/\log\log N)}$, ruling out efficient SQ algorithms learning intersections of $\omega(\log \log N)$ halfspaces.
As a by-product of our results, we will give an improved SQ lower bound (via a different hard instance than~\cite{klivans2007unconditional}), showing that learning intersections of $k$ halfspaces needs precision at least $N^{-\Omega(k)}$.
Note that this rules out efficient SQ algorithms for learning intersections of $\omega(1)$ halfspaces but also gives a fine-grained hardness result for $k = O(1)$.

\paragraph{Hardness assumption, SQ model and main results}
We will next state the precise hardness assumption we make.
We remark that we do not expect the reader to be familiar with lattices or these problems and such familiarity is not necessary in order to understand and appreciate the remainder of this paper.
Our reductions will start from a different learning problem that can be stated in elementary terms (see~\cref{sec:techniques}).
For more background on lattices and these problems, we refer to \cite{peikert2016decade}.

An $n$-dimensional \emph{lattice} $L$ is defined to be a discrete additive subgroup of $\R^n$.
It can be fully specified by a basis $B \in \R^{n \times n}$ as $L = B \Z^{n}$.
We will only consider the case in which $B$ is full-rank.
For $1\leq i \leq n$, consider 
\[
    \lambda_i\Paren{L} \coloneqq \inf\Set{r > 0 \suchthat \dim\Paren{\Span\Paren{L \cap B_r(0)} \geq i}} \,.
\]
We can now define $\mathrm{GapSVP}$ and $\mathrm{SIVP}$.
\begin{problem}[Gap Shortest Vector Problem ($\mathrm{GapSVP}$)]
    \label{prob:gap_svp}
    Let $\alpha = \poly(n)$ be arbitrary.
    Given an $n$-dimensional lattice $L$ and $d > 0$ such that either (a) $\lambda_1\Paren{L} \leq d$ or (b) $\lambda_1\Paren{L} > \alpha \cdot d$, decide whether (a) or (b) holds.
\end{problem}

\begin{problem}[Shortest Independent Vector Problem ($\mathrm{SIVP}$)]
    \label{prob:sivp}
    Let $\alpha = \poly(n)$ be arbitrary.
    Given an $n$-dimensional lattice $L$ output a set of linearly independent lattice points of length at most $\alpha \cdot \lambda_n\Paren{L}$.
\end{problem}

We make the following assumption
\begin{assumption}
    \label{assump:lattices}
    There is no quantum algorithm that runs in time $2^{o(n)}$ and uses only $2^{o(n)}$ samples that solves either \cref{prob:gap_svp} or \cref{prob:sivp}.
\end{assumption}
All known (quantum) algorithms for \cref{prob:gap_svp} and \cref{prob:sivp} require time $2^{\Omega(n)}$.
Further, a falsification of the above assumption would be considered a major breakthrough in cryptography (cf.~\cite{peikert2016decade} and references therein for more context).

Similarly, we give some necessary background on the SQ model.
In particular, SQ algorithms only have access to the distributions via query functions $\phi \colon \R^N\times \Set{\pm 1} \rightarrow [-1,1]$.
Upon making a query $\phi$, they receive as an answer a value in $[\E_D \phi(x,y) - \tau,\E_D \phi(x,y) + \tau]$.
$\tau$ is called the \emph{accuracy} or \emph{precision} of the query.
The query function can be arbitrary and outside of making these queries, the algorithms can perform arbitrary computation.
When comparing to sample-based algorithms, typically the number of queries is taken as a proxy for run-time and $1/\tau^2$ as a proxy for the number of samples -- since this many samples are needed to estimate the expectation of a query from samples up to accuracy $\tau$.

Our reduction-based hardness result is as follows
\begin{theorem}[See~\cref{thm:main_crypto_hardness} for full version]
    \label{thm:crypto_intro}
    Let $N, k \in \N$ such that $k \leq O(\sqrt{N})$.
    Under \cref{assump:lattices}, there is no $T = N^{o(\tfrac{k}{\log k + \log\log N})}$-time algorithm that learns the intersection of $k$ halfspaces up to error better than $\tfrac 1 2 - \tfrac 1 {\Omega(T)}$.
\end{theorem}
It is insightful to explicitly compute the time lower bound for specific values of $k$.
First, note that this rules out polynomial-time algorithms for weakly learning the intersection of $\omega(\log \log N)$ halfspaces.
A few other examples are as follows:
For any $0 < \e \leq \tfrac 1 2$, not necessarily constant, learning intersections of $k = N^\e$ halfspaces takes time at least $\exp(\Omega(N^\e \cdot \tfrac{\log N}{\e \log N + \log\log N}))$.
In particular, taking $\e$ to be an absolute constant, we obtain that learning intersections of $N^\e$ halfspaces takes time at least $\exp(\Omega(N^\e))$.
Taking $\e = \tfrac{\log\log N}{\log N}$, we obtain that learning intersections of $\log N$ halfspaces takes time $\exp(\Omega(\tfrac{\log^2 N}{\log \log N}))$.
Finally, for $\e = \omega(\tfrac{\log\log\log N}{\log N})$, we recover the result for $k = \omega(\log \log N)$ mentioned in the beginning.

Further, under the more conservative assumption that there is no algorithm for~\cref{prob:gap_svp,prob:sivp} running in time $2^{\Omega(n^{1-\delta})}$ for any constant $\delta > 0$, we are still able to rule out polynomial-time algorithms for weakly learning the intersection of $\omega(\log^\delta (N))$ halfspaces.
See~\cref{sec:relation_prev_work} for a more detailed comparison with prior work.

Interestingly, our lower bound instance also satisfies a margin property.
The time lower bound we obtain comes close to the runtime of existing algorithms exploiting such a margin assumption \cite{learning_intersections_margins}.
See \cref{sec:crypto_hardness} for a more detailed discussion.

Our SQ hardness results are as follows:
\begin{theorem}
    \label{thm:sq_intro}
    Let $k,N \in \N$ such that $k \leq N^\gamma$ for a sufficiently small absolute constant $\gamma$.
    Any SQ algorithm using queries of accuracy $\tau = N^{-\Omega(k)}$ that learns the intersection of $k$ halfspaces over $\R^N$ up to error better than $\tfrac 1 2 - 4\tau$ must make at least $2^{N^{\Omega(1)}}$ queries.
\end{theorem}
Note that this shows that even weakly learning intersections of $\omega(1)$ halfspaces requires super-polynomial precision in the SQ model or exponentially many queries.
Similarly, it shows that the fine-grained complexity of learning intersections of $k$ halfspaces scales as $N^{\Omega(k)}$.
We remark that we prioritized clarity and did not attempt to optimize any constants, neither in the condition that $k \leq N^\gamma$ nor in the exponent of the accuracy or the number of queries.

\paragraph{Future work}

We remark that both our lower bound instance can be solved in time $N^{O(k)}$ since they can be represented as a degree-$O(k)$ polynomial threshold function (see~\cref{sec:techniques} for all details) -- and thus can be learned in time $N^{O(k)}$ via linear programming~\cite{maass1994fast}.
This suggests that we should look for instances that cannot be represented as low-degree polynomial threshold functions.
This approach seems particularly well-motivated since it is known that, at least when the input comes from the boolean hypercube, there exists an intersection of even 2 halfspaces that cannot be represented by degree-$o(n)$ polynomial threshold functions~\cite{sherstov2010optimal,sherstov2021hardest}.

\subsection{More on Previous Results}
\label{sec:relation_prev_work}

We elaborate a bit more on the connection between our work and previous hardness results below.

The work closest to us is~\cite{klivans2009cryptographic} (and the companion work~\cite{klivans2006cryptographic}).
Their hardness results are ultimately also based on the hardness of~\cref{prob:gap_svp,prob:sivp}.
However, their hardness result follows by showing that intersections of halfspaces can encode a public-key encryption system due to Regev~\cite{DBLP:journals/jacm/Regev09} known to be secure assuming hardness of these lattice problems.
Thus, a learning algorithm could break the crypto-system and hence falsify~\cref{prob:gap_svp,prob:sivp}.
While we start from the same assumptions, we give a more direct reduction, completely bypassing the need to introduce any public-key encryption schemes.
This more direct reduction is what enables our improved SQ lower bounds.

On a quantitative level,~\cite{klivans2009cryptographic} shows that for any absolute constant $\e > 0$ a $\poly(N)$-time algorithm for learning intersections of $k = N^\e$ halfspaces in dimension $N$ would yield a $\poly(n)$-time algorithm for~\cref{prob:gap_svp,prob:sivp} in dimension $n$.
In particular, their results are implied by a weaker version of~\cref{assump:lattices} in which we only assume that there is no $\poly(n)$-time algorithm for~\cref{prob:gap_svp,prob:sivp}.\footnote{More specifically, they show that is true even when setting $\alpha = \tilde{O}(n^{1.5})$. We strongly believe that this is also true for our reduction but we did not attempt to make this explicit for clarity. ~\cref{prob:gap_svp,prob:sivp} are believed to be hard for any $\alpha = \poly(n)$.}
In~\cite{klivans2006cryptographic}, the same authors observed that their reduction implies stronger lower bounds under quantitatively stronger assumptions on~\cref{prob:gap_svp,prob:sivp} (closer to our~\cref{assump:lattices}).
Pushed to the limit, their result yields that~\cref{assump:lattices} implies that learning intersections of $\omega(\log N)$ halfspaces in dimension $N$ takes super-polynomial time, matching the result of~\cite{daniely2016complexity} under a different assumption -- we remark that this is not formally stated in~\cite{klivans2006cryptographic} but follows immediately from their techniques.

In particular, allowing $\e$ to be sub-constant, their techniques can be used to  show that~\cref{assump:lattices} implies that learning intersections of $N^\e$ halfspaces takes time $\exp(\Omega(N^\e))$ (see the discussion at the end of~\cref{sec:techniques} for a more detailed argument and technical comparison to our work).
This should be compared with our lower bound $\exp(\Omega(N^\e \tfrac{\log N}{\log\log N}))$ for $\e \leq \tfrac{\log\log N }{\log N}$ (and similar for larger $\e$).
This lower bound is significantly larger and in particular allows to obtain hardness results of exponentially fewer halfspaces ($\omega(\log\log N)$).

We strongly believe that our techniques also allow for a trade-off of the following form:
To rule out polynomial-time algorithms for learning intersections of more halfspaces under quantitatively weaker assumptions.
We choose not to make this explicit for clarity of exposition and since already a $2^{o(n)}$-time algorithm for either of~\cref{prob:gap_svp,prob:sivp} would be a major breakthrough.

Along a different direction,~\cite{tiegel_agnostic_lattices,diakonikolas_massart_lwe,diakonikolas2023near} show lower bounds for learning a single halfspaces in various error models under \cref{assump:lattices}.

\section{Technical Overview}
\label{sec:techniques}

\paragraph{Relation to parallel pancakes and SQ lower bound}

Our lower bounds are based on a novel connection we make between the so-called ``parallel pancakes'' distribution~\cite{DKS17,bubeck2019adversarial,CLWE} and intersections of halfspaces.
On a high level, the former is a mixture of few Gaussians, that is hard to distinguish from the standard Gaussian distribution.
It (or versions thereof) has played a pivotal role in obtaining computational hardness results for learning theory problems.
Yet, the connection to intersections of halfspaces had not been observed before.
Similar ideas, without any reference to parallel pancakes, were implicitly used in~\cite{klivans2009cryptographic}.
By making this connection explicit and expanding on it, we are able to obtain improved lower bounds in both the SQ model and under \cref{assump:lattices}.
More specifically, our connection allows us to leverage that (variants of this) distribution are known to be hard to learn in the SQ model and based on \cref{assump:lattices}.
Fleshing out all details and satisfying all distribution requirements exactly will take some additional work.

We start by describing one version of the parallel pancakes distribution and showing our SQ lower bound (\cref{thm:sq_intro}).
Unfortunately, this connection alone is not enough to establish our reduction-based result (\cref{thm:crypto_intro}) as well.
The reason is that the known hardness results for parallel pancakes under \cref{assump:lattices} are quantitatively weaker than those known under SQ -- and in particular would by themselves only rule out efficient algorithms for learning $\omega(\log N)$ halfspaces.
Towards the end of this section, we will outline how to show a hardness result for $\omega(\log \log N)$ halfspaces using a modified construction.

It is known~\cite{bubeck2019adversarial} (see also~\cite{DKS17}) that there are two one-dimensional distributions $A,B$ satisfying the following properties (see~\cref{sec:sq_hardness} for all details):
\begin{enumerate}
    \item $A$ and $B$ are mixtures of $k$ Gaussians.
    \item There exists two unions of $k$ disjoint intervals $S_A$ and $S_B$, such that only a negligible fraction of the probability mass of $A$ (resp. $B$) lies outside $S_A$ (resp. $S_B$).
    \item The intervals in $S_A \cup S_B$ are disjoint and ``interlacing'' in the sense that they alternate.
    \item Both $A$ and $B$ match $k$ moments with $N(0,1)$.
\end{enumerate}
Consider now the following distribution $D_{A,B}$ over $\R^N\times\Set{\pm 1}$: First, pick a uniformly random unit vector $w$, and let $D_A$ (resp. $D_B$) be the distribution over $\R^N$ that is $A$ (resp. $B$) along $w$ and a standard Gaussian in the complement.
Then, set $D_{A,B} = \tfrac 1 2 (D_A, +1) + \tfrac 1 2 (D_B, -1)$.
Using results from~\cite{bubeck2019adversarial,DK22} it is not hard to deduce that $D_{A,B}$ is hard to distinguish from $N(0,\Id_N) \times \Be{\tfrac 1 2}$ in the SQ model.
In particular, this task either requires queries of accuracy better than $N^{-\Omega(k)}$ (suggesting that we need at least $N^{\Omega(k)}$ samples) or $2^{N^{\Omega(1)}}$ queries.
(We give a full argument for our variant of this distribution in~\cref{sec:sq_hardness}.)

We offer some intuition on $D_{A,B}$:
This distribution corresponds to a mixture of two related instances of a labelled version of the parallel pancakes distribution alluded to earlier.
In particular, for $(x,y) \sim D_{A,B}$, if $y = +1$, $x$ follows the standard parallel pancakes distribution and if $y = -1$, $x$ follows a ``shifted'' parallel pancakes distribution in which the pancakes are shifted along the hidden direction such that they are (mostly) disjoint from the pancakes for $y = +1$.
See Figure 1 for an illustration.

\paragraph{Modifying the instance to obtain an intersection of degree-2 PTFs}
Our SQ lower bound follows from the simple but powerful observation that a slight variant of this distribution can be realized as an intersection of $k$ degree-2 polynomial threshold functions (denoted PTFs for short).
Note that this is enough to show our hardness result.
Indeed, recall that we aim to show that learning the intersection of $k$ halfspaces in dimension $N$ takes time $N^{\Omega(k)}$.
For this it is sufficient to show that learning the intersection of $k$ degree-2 PTFs in dimension $N$ takes time at least $N^{\Omega(k)}$ since we can represent these as intersections of $k$ halfspaces over an $O(N^2)$-dimensional space.
We can absorb the quadratic blow-up in the dimension in the $\Omega(\cdot)$-notation.

Note that a priori $D_{A,B}$ cannot be realized as such an intersection:
Since the density of both $A$ and $B$ are positive on all of $\R$, there exists a region in which the label can be both + and -1 with some small probability.
Since our model is noiseless, this should not be possible.
Fortunately, these regions only make up a small fraction of the total probability mass and we can get rid of them by truncating the mixture components.
Indeed, let $\tilde{A},\tilde{B}$ be obtained by conditioning $A$ (resp. $B$) to lie in $S_A$ (resp. $S_B$) and let $D_{\tilde{A},\tilde{B}}$ be obtained analogously as before (replacing $A$ and $B$ by $\tilde{A}$ and $\tilde{B}$).
In~\cref{sec:sq_hardness} we show that $D_{\tilde{A},\tilde{B}}$ enjoys the same hardness guarantees in the SQ model as $D_{A,B}$, i.e, that this distribution is still hard to distinguish from $N(0,\Id_N) \times \Be{\tfrac 1 2}$ in the relevant parameter regime.
This follows by showing that the first $k$ moments of both $\tilde{A}$ and $\tilde{B}$ still match those of $N(0,1)$ up to small error ($N^{-\Omega(k)}$) and their $\chi^2$-divergence with $N(0,1)$ is not too large ($2^{O(k)} \log N$) -- this uses results based on~\cite{DK22}.

To see that $D_{\tilde{A},\tilde{B}}$ is an intersection of $k$ degree-2 PTFs, note the following:
By construction, for a sample $(x,y) \sim D_{\tilde{A},\tilde{B}}$, $y = 1$ if and only if $\iprod{x,w} \in S_A$.
Further $y = -1$ if and only if $\iprod{x,w} \in S_B$.
Thus, for every interval $I \sse S_B$, consider the polynomial $p_I \colon \R \rightarrow \R$ that is symmetric around the mid-point of $B$, is negative on $I$, and has its roots at half the distance between the end of $I$ and the next interval in $S_A$.
Note that $p_I$ is negative on $I$ and positive on all other intervals in both $S_A$ and $S_B$.
The final choice of degree-2 PTFs is then $\tilde{p}_I \colon \R^N \rightarrow \R$ such that $\tilde{p}_I(x) = p_I(\iprod{x,w})$.
By construction, if $\iprod{x,w} \in S_A$, $\tilde{p}_I(x) \geq 0$ for all $I$ and if $\iprod{x,w} \in S_B$ there exists $\tilde{p}_I$ such that $\tilde{p}_I(x) < 0$.
It follows that $D$ corresponds to the intersection of the $\tilde{p}_I$.
Since $S_B$ contains $k$ intervals, this yields the claim.
See Figure 1 (a) for an illustration.

\definecolor{ao(english)}{rgb}{0.0, 0.5, 0.0}

\begin{figure}
\RawFloats
  \begin{minipage}[b]{0.45\linewidth}
    \centering
    \pgfplotsset{samples=100}

\begin{tikzpicture}
  \begin{axis}[
    width=8cm,
    height=6cm,
    axis lines=middle,
    xlabel=$w$,
    ylabel=,
    xmin=-1.5,
    xmax=1.5,
    ymin=-0.2,
    ymax=1.1,
    domain=-3.5:3.5,
    xtick={-1,-0.5,0,0.5,1},
    xticklabels={},
    xticklabel shift={10pt},
    ytick=\empty,
    ticklabel style={font=\scriptsize},
    legend style={at={(0.97,0.9)},anchor=north east,font=\tiny},
    legend cell align=right
  ]
  
  \def\gam{2};
  \def\bet{0.2};
  \def\delt{0.25 + 1} 
    \addplot[blue,smooth,line width=1pt,domain=-0.2:0.2] { exp(-pi*x^2) * exp(-pi * ( (0 - \gam * x) / \bet )^2 ) };
  \addlegendentry{$y=+1$};

  \addplot[red,smooth,line width=1pt,domain=0.3:0.7] { exp(-pi*x^2) * exp(-pi * ( (1 - \gam * x) / \bet )^2 ) };
  \addlegendentry{$y=-1$};

  \addplot[ao(english),dashed, line width=1pt, restrict y to domain=-3.5:0.2] { 3 * (x-0.25) * (x-0.75)};
  \addplot[ao(english),dashed, line width=1pt, restrict y to domain=-3.5:0.2] { 3 * (x+0.25) * (x+0.75)};
  \addplot[ao(english),dashed, line width=1pt, restrict y to domain=-3.5:0.2] { 3 * (x+1.75) * (x+1.25)};
  \addplot[ao(english),dashed, line width=1pt, restrict y to domain=-3.5:0.2] { 3 * (x+0.25) * (x+0.75)};
  \addplot[ao(english),dashed, line width=1pt, restrict y to domain=-3.5:0.2] { 3 * (x-1.25) * (x-1.75)};
  \addlegendentry{Collection of $p_I$'s};

  \addplot[blue,smooth,line width=1pt,domain=0.8:1.2] { exp(-pi*x^2) * exp(-pi * ( (2 - \gam * x) / \bet )^2 ) };
  \addplot[blue,smooth,line width=1pt,domain=-1.2:-0.8] { exp(-pi*x^2) * exp(-pi * ( (-2 - \gam * x) / \bet )^2 ) };

  \addplot[red,smooth,line width=1pt,domain=-0.7:-0.3] { exp(-pi*x^2) * exp(-pi * ( (-1 - \gam * x) / \bet )^2 ) };

  \addplot[red,smooth,line width=1pt,domain=1.3:1.7] { 20*exp(-pi*x^2) * exp(-pi * ( (3 - \gam * x) / \bet )^2 ) };
  \addplot[red,smooth,line width=1pt,domain=-1.7:-1.3] { 20*exp(-pi*x^2) * exp(-pi * ( (-3 - \gam * x) / \bet )^2 ) };
  
  \end{axis}
\end{tikzpicture}

    \caption*{(a)}
  \end{minipage}%
  \hspace{0.09\linewidth}
  \begin{minipage}[b]{0.45\linewidth}
    \centering
    
    \pgfplotsset{samples=100}

\begin{tikzpicture}
  \begin{axis}[
    width=8cm,
    height=6cm,
    axis lines=middle,
    xlabel=$w$,
    ylabel=,
    xmin=-1.5,
    xmax=1.5,
    ymin=-0.2,
    ymax=1.1,
    domain=-3.5:3.5,
    xtick={-1,-0.5,0,0.5,1},
    xticklabels={},
    xticklabel shift={10pt},
    ytick=\empty,
    ticklabel style={font=\scriptsize},
    legend style={at={(0.97,0.9)},anchor=north east,font=\tiny},
    legend cell align=right
  ]
  
  \def\gam{2}
  \def\bet{0.2}
  \def\delt{0.25 + 1} 
    \addplot[blue,smooth,line width=1pt,domain=-0.2:0.2] { exp(-pi*x^2) * exp(-pi * ( (0 - \gam * x) / \bet )^2 ) };
  \addlegendentry{$y=+1$};

  \addplot[red,smooth,line width=1pt,domain=0.3:0.7] { exp(-pi*x^2) * exp(-pi * ( (1 - \gam * x) / \bet )^2 ) };
  \addlegendentry{$y=-1$};

  \addplot[orange,dashed, line width=1pt, restrict y to domain=-3.5:0.4] { 3 * (x-0.25) * (x-0.75) * (x-1.25) * (x-1.75)};
  \addplot[orange,dashed, line width=1pt, restrict y to domain=-3.5:0.4] { 3 * (x+0.25) * (x+0.75) * (x+1.75) * (x+1.25)};

  \addlegendentry{Collection of $p_I$'s};

  \addplot[blue,smooth,line width=1pt,domain=0.8:1.2] { exp(-pi*x^2) * exp(-pi * ( (2 - \gam * x) / \bet )^2 ) };
  \addplot[blue,smooth,line width=1pt,domain=-1.2:-0.8] { exp(-pi*x^2) * exp(-pi * ( (-2 - \gam * x) / \bet )^2 ) };

  \addplot[red,smooth,line width=1pt,domain=-0.7:-0.3] { exp(-pi*x^2) * exp(-pi * ( (-1 - \gam * x) / \bet )^2 ) };

  \addplot[red,smooth,line width=1pt,domain=1.3:1.7] { 20*exp(-pi*x^2) * exp(-pi * ( (3 - \gam * x) / \bet )^2 ) };
  \addplot[red,smooth,line width=1pt,domain=-1.7:-1.3] { 20*exp(-pi*x^2) * exp(-pi * ( (-3 - \gam * x) / \bet )^2 ) };

  \end{axis}
\end{tikzpicture}

    \caption*{(b)}
  \end{minipage}
  \caption{(a) shows how to capture the parallel pancakes distribution using degree-2 PTFs and (b) shows how to do the same using higher-degree PTFs (degree-4 in this case).}
    \label{fig:pancakes}
\end{figure}

To solve the distinguishing problem, we can run our weak learner on our input distribution and with one additional query compute the misclassification error of the produced hypothesis.
Since in the null case the label $y$ is independent of $x$, this should be 1/2, while it should be bounded away from $1/2$ in the planted case by assumption on our weak learner.
We can thus solve the distinguishing problem.

\paragraph{Lower bound based on \cref{assump:lattices}}

``Parallel Pancakes''-type distributions are also known to be hard to distinguish from a standard Gaussian under \cref{assump:lattices}.
In particular, using results from~\cite{CLWE,CLWE_2} one could show that a similar distribution, that also has $k$ components, takes time roughly at least $2^{\Omega(k)}$ to distinguish from a standard Gaussian.
Unfortunately, using this, we could only hope to rule out learning intersections of $\omega(\log N)$ halfspaces, which is exponentially worse than $\omega(\log \log N)$.
In order to obtain our improved lower bound, we make use of the following observation:
Instead of showing that intersections of degree-2 PTFs are hard to learn, we can also show that degree-$d$ PTFs are hard to learn for $d > 2$.
Note that this introduces a fundamental tradeoff:
The larger we choose $d$, the smaller the number of halfspaces becomes but the blow-up in the dimension is exponential in $d$.
Luckily for us, there is still a choice of $d$ that rules out learning $\omega(\log \log N)$ halfspaces.

\cite{tiegel_agnostic_lattices} (building on~\cite{CLWE}) showed the following (see~\cref{sec:crypto_hardness} for all details\footnote{\cite{tiegel_agnostic_lattices} used a construction based on these distributions to show that learning a single halfspace in the agnostic model is hard under \cref{assump:lattices}. Note that this is different from our setting as we do not allow noise in the labels.}): There are two one-dimensional distributions $A,B$ satisfying
\begin{enumerate}
    \item $A,B$ are mixtures of infinitely many (truncated) Gaussians.
    \item There exists two unions of infinitely many disjoint intervals $S_A,S_B$, such that $A$ (resp. $B$) is supported on $S_A$ (resp. $S_B$).
    \item The intervals in $S_A \cup S_B$ are disjoint and ``interlacing'' in the sense that they alternate.
    \item If there is an algorithm distinguishing $D_{A,B}$ from $N(0,\Id_n) \times \Be{\tfrac 1 2}$ using $2^{o(n)}$ samples and running in time $2^{o(n)}$, then \cref{assump:lattices} is false.
\end{enumerate}
In what follows we will denote the dimension of $D_{A,B}$ by $n$.
We will denote the dimension of the halfspaces by $N$ (which will roughly be $n^d$).
Our first observation is that we can restrict to the $2n+1$ most central intervals in $A$ and $B$ respectively.
It is not hard to show that the resulting $D_{A,B}$ is $2^{-\Omega(n)}$-close to the original one in total variation distance.
Thus, even when seeing $2^{o(n)}$ samples from this distribution, the respective product distributions are still $2^{-\Omega(n)}$ close in total variation, and hence, the associated distinguishing problem is just as hard.
We can hence assume that $S_A, S_B$ contain only $2n+1$ intervals.

Let $d = \tfrac {2n+1} k + 1$ and for simplicity assume this is an even integer.
By a similar construction as for the SQ lower bound, $D_{A,B}$ can be realized as an intersection of $k$ degree-$d$ PTFs -- this time each PTF traces out $d-1$ intervals in $S_B$, instead of just 1.
See \Cref{fig:pancakes}(b) for an illustration. 
These can be realized as an intersection of $k$ halfspaces in dimension $N = \Theta(n^d)$.
Recall that we want to rule out algorithms weakly learning the intersections of halfspaces that run in time $N^{o(\tfrac k {\log k +\log \log N})}$.
We claim that such an algorithm can distinguish $D_{A,B}$ from $N(0,\Id_n)\times \Be{1/2}$.
In particular, since $\log N = \Theta(d \cdot \log n) = \Theta(\tfrac n k \cdot \log n)$ an algorithm running in time $N^{o(\tfrac k {\log k +\log \log N})}$ runs in time $2^{o(n)}$.
Indeed,
\begin{align*}
    N^{o\Paren{\tfrac{k}{\Paren{\log k + \log\log N}}}} &= \exp\Paren{o\Paren{\frac {k \cdot \log N}{\log k + \log\log N}}} = \exp\Paren{o\Paren{\frac {n \cdot \log n}{\log k + \log n - \log k + \log\log n}}} \\
        &=2^{o\Paren{n}}\,.
\end{align*}
Thus, to solve the distinguishing problem we can use a similar reduction as in the SQ model:
Run the learner on the first half of the input samples and compute the empirical misclassification error on the second.
Again, under null this should be very close to 1/2 whereas under planted it should be bounded away from 1/2.

\paragraph{Comparison to~\cite{klivans2009cryptographic}}

The work~\cite{klivans2009cryptographic} shows that the intersection of $O(n)$ degree-2 PTFs can encode the decryption function of a crypto-sytem by Regev~\cite{DBLP:journals/jacm/Regev09}.
Under \cref{assump:lattices} breaking this crypto-system requires time at least $2^{\Omega(n)}$.
Using a similar argument as above, they deduce that learning $O(\sqrt{N})$ halfspaces in dimension $N$ takes time at least $2^{\Omega(\sqrt{N})}$ -- where the $\sqrt{N}$ comes from the quadratic blow-up in the dimension.
Further, they argue the following: For any $\e > 0$, by padding all vectors with 0, we can artificially blow-up the dimension to $N = n^{\tfrac 1 \e}$.
The number of halfspaces is then $k = N^\e$ (over the $N$-dimensional space) and the learning task requires time at least $2^{\Omega(n)} = \exp(N^\e) = \exp(k)$.
It follows that learning intersections of $\omega(\log N)$ halfspaces in dimension $N$ takes time at least $N^{\omega(1)}$.

Note that this simple padding argument cannot go beyond $\omega(\log N)$ halfspaces, intuitively, the padding strategy does not exploit the additional space available in higher dimensions.
On the other hand, our argument based on higher-degree PTFs shows that exploiting this is indeed possible.
Further, our arguments completely bypass the need to introduce any crypto-systems.
In fact, it is not clear how the construction based on Regev's crypto-system would yield unconditional lower bounds in the SQ model.

\section{Preliminaries}
\label{sec:preliminaries}

\paragraph{Notation}
We denote $\R_{\geq 0} = [0,\infty)$ and $\R_{> 0} = (0,\infty)$.
For a set $S$, we denote by $\cU(S)$ the uniform distribution over $S$.
We define the Total Variation Distance between two measures $P$ and $Q$ as \[\TVD{P}{Q} = \sup_A \Abs{P(A) - Q(A)}\,.\]

Let $n$ be some parameter.
For the problem of distinguishing two distributions $D_n^{0}$ and $D_n^1$ we define the advantage of an algorithm $\cA$ as \[\Abs{\Psymb_{x \sim D_n^0} \Paren{\cA(x) = 0} - \Psymb_{x \sim D_n^1} \Paren{\cA(x) = 0}}\,.\]
We say that an algorithm has non-negligible advantage if it has advantage $\Omega(n^{-c})$ for some constant $c > 0$.

Let $p \in [0,1/2]$.
We denote by $\mathrm{Be}(p)$ the distribution that is equal to +1 with probability $p$ and equal to -1 with probability $1-p$.

Let $\cX$ be some set and $D$ be a distribution over $\cX \times \Set{-1,+1}$.
Further, let $h \colon \cX \rightarrow \Set{-1,+1}$ be a binary hypothesis.
We denote the \emph{misclassification error} of $h$ as \[\err_D \Paren{h} = \Psymb_{(x,y) \sim D}\Paren{h(x) \neq y} \,.\]
Most of the time the distribution $D$ will be clear from context and we will omit the subscript.
We denote by $D_x$ the marginal distribution of $D$ over $\cX$.
If the domain of $D_x$ is $\R^n$, we say an algorithm weakly learns $D$, if it outputs a binary hypothesis $\hat{h}$ such that $\err_D\paren{\hat{h}} \leq \tfrac 1 2 - \tfrac 1 {\poly(n)}$ for some choice of $\poly(n)$.

\paragraph{Margin}
Let $c \colon \R^n \rightarrow \Set{-1,+1}$ be any boolean function and $X \sse \R^n$.
We say $c$ has margin $\rho$ with respect to $X$, if
\[
    \frac{\min \Set{\norm{x - z} \suchthat x \in X \,, z \in \R^n \,, c(z) \neq c(x) } } {\sup_{x \in X} \norm{X}} \geq \rho \,.
\]
For a distribution $D_x$ over $\R^n$ we say that $c$ has margin $\rho$ with respect to $D_x$, if it has margin $\rho$ with respect to the set of points that have strictly positive density under $D_x$.
Note that in our definition of margin we normalize the points to have at most unit norm.\footnote{For intersections of halfspaces we could equivalently, up to a factor of 2, define the margin to be the smallest distance of a point of non-zero probability mass to one of the halfspaces.}

\paragraph{Gaussian distributions}
We denote the standard $n$-dimensional Gaussian distribution by $N(0,I_n)$.
If the dimension is clear from context, we sometimes drop the subscript of the identity matrix.
For $s > 0$, we denote by $\rho_s \colon \R^n \rightarrow \R_+$ the function \[\rho_s(x) = \exp(-\pi \snorm{x /s})\,.\]
If $s = 1$, we omit the subscript.
Note that $\rho_s/s^n$ is equal to the probability density function of the $n$-dimensional Gaussian distribution with mean 0 and covariance matrix $s^2/(2\pi) \cdot I_n$.
In particular, it holds that \[\int_{\R^n} \rho_s(x) \,dx = s^n \,.\]
For a radius $\alpha > 0$ we define the following truncated version of $\rho_s$.
\[\rho_s^{\alpha}(x) = \begin{cases}
    \tfrac 1 Z \cdot \rho_s(x) \,, &\quad \text{if } \Norm{x} \leq \alpha \,, \\
    0 \,, &\quad \text{otherwise,}
\end{cases}\]
where \[Z = \frac{\int_{\Norm{x} \leq \alpha} \rho_s(x) \, d x}{\int_{\R} \rho_s(x) \, d x}\,.\]

For a lattice $L \sse \R^n$ and $s > 0$ we define the discrete Gaussian distribution $D_{L, s}$ with width $s$ as having support $L$ and probability mass proportional to $\rho_s$.
Further, for a discrete set $S$, we define $\rho_s\Paren{S} = \sum_{x \in S} \rho_s\Paren{x}$.

\paragraph{Various versions of Continuous LWE}

\begin{definition}[CLWE Distribution]
    \label{def:clwe_distribution}
    Let $w \in \R^n$ be a unit vector and $\beta, \gamma > 0$.
    Define the distribution $\clwed{w}{\beta}{\gamma}$ over $\R^n \times [0,1)$ as follows.
    Draw $y \sim N(0,\tfrac 1 {2\pi} \cdot I_n)$, $e \sim N(0,\beta^2/(2\pi))$ and let \[z = \gamma \cdot \iprod{y,w} + e \mod 1\,.\]
    Note that the density of this distribution is given by \[p(y, z) = \frac{1}{\beta} \cdot \rho \Paren{y} \cdot \sum_{k \in \Z} \rho_\beta \Paren{z + k - \gamma \iprod{w, y}} \,.\]

    Further, let $m \in\N$.
    We denote by $\clwem{m}{\gamma}{\beta}$ the distribution obtained by first drawing $w \sim \cU(\cS^{n-1})$ and then drawing $m$ independent samples from $\clwed{w}{\gamma}{\beta}$.
\end{definition}

\begin{definition}[Homogeneous CLWE (hCLWE) Distribution]
    \label{def:hclwe_distribution}
    Let $w \in \R^n$ be a unit vector, $c \in [0,1)$, and $\beta, \gamma > 0$.
    Let $\pi_{{w}^\perp}(y)$ be the projection of $y$ onto the space orthogonal to $w$.
    Define the distribution $\hclwed{w}{\beta}{\gamma}{c}$ over $\R^n$ as having density at $y$ proportional to
    \begin{align}
        \sum_{k \in \Z} \rho_{\sqrt{\beta^2 + \gamma^2}}(k - c) \cdot \rho\Paren{\pi_{{w}^\perp}(y)} \cdot \rho_{\beta / \sqrt{\beta^2 + \gamma^2}}\Paren{\iprod{w, y} - \tfrac{\gamma}{\beta^2 + \gamma^2} (k-c) }\,. \label{eq:hclwe}
    \end{align}

    Further, let $m \in\N$.
    We denote by $\hclwem{m}{\gamma}{\beta}{c}$ the distribution obtained by first drawing $w \sim \cU(\cS^{n-1})$ and then drawing $m$ independent samples from $\hclwed{w}{\gamma}{\beta}{c}$.
\end{definition}
Intuitively, one can think of the $\hclwed{w}{\gamma}{\beta}{c}$ distribution as $\clwed{w}{\gamma}{\beta}$ conditioned on $z = c$.

\begin{definition}[Truncated hCLWE Distribution]
    Let $w \in \R^n$ be a unit vector, $c \in [0,1), \beta, \gamma > 0$ and $\alpha = \frac 1 {10} \cdot \frac{\gamma}{\gamma^2 + \beta^2}$.
    Define the distribution $\nhclwed{w}{\beta}{\gamma}{c}^{(n)}$ over $\R^n$ as having density proportional to
    \begin{align}
        \sum_{k = -n}^n \rho_{\sqrt{\beta^2 + \gamma^2}}(k - c) \cdot \rho\Paren{\pi_{{w}^\perp}(y)} \cdot \rho_{\beta / \sqrt{\beta^2 + \gamma^2}}^\alpha\Paren{\iprod{w, y} - \tfrac{\gamma}{\beta^2 + \gamma^2}\Paren{k - c} }\,. \label{eq:nhclwe}
    \end{align}
    The superscript refers to the range of the summation.

    Further, let $m \in\N$ and $\cS$ be a distribution over unit vectors in $\R^n$.
    We denote by $\nhclwem{m}{\gamma}{\beta}{c}$ the distribution obtained by first drawing $w \sim \cU(\cS^{n-1})$ and then drawing $m$ independent samples from $\nhclwed{w}{\gamma}{\beta}{c}$.
\end{definition}

Note that this is the same as the hCLWE distribution but with the individual components of the mixture truncated in the hidden direction and restricting to the middle $2n+1$ components.
$\alpha$ is chosen such that the components become non-overlapping but the resulting distribution has small total variation distance to the corresponding non-truncated hCLWE distribution.
Although this is strictly speaking not necessary to prove our result, we will see that having non-overlapping components will simplify our analysis.

We make the following hardness assumption about the CLWE distribution.
\begin{assumption}
    \label{assump:disting_clwe}
    Let $n, m \in \N$ and \[\gamma \geq 2\sqrt{n}\,,\quad\quad \beta = \frac 1 {\poly\Paren{n}}\,.\]
    Further, let $\delta < 1$ be arbitrary and $m = 2^{n^\delta}$.
    There is no $2^{n^\delta}$-time distinguisher between \[\clwem{m}{\gamma}{\beta} \quad\text{and}\quad N\Paren{0, \tfrac 1 {2\pi} \cdot I_n}^m \times U\Paren{[0,1)}^m\] with non-negligible advantage.
\end{assumption}

Note that by \cite[Corollary 3.2]{CLWE} this is implied by assuming~\cref{assump:lattices}.

\paragraph{Hermite polynomials and moment-matching distributions}

We will also use the following facts about one-dimensional distributions matching moments with $N(0,1)$.
\begin{fact}[\cite{bubeck2019adversarial}]
    \label{fact:moment_matching_hermite}
    For every $k \in \N$ greater or equal to 2, there exist two discrete distributions $A$ and $B$ supported on at most $k$ points such that 
    \begin{itemize}
        \item $A$ and $B$ match at least $2k-1$ moments with $N(0,1)$,
        \item The points in the union of the supports of $A$ and $B$ are pairwise at distance at least $\Omega(1/\sqrt{k})$. Further, they are all contained in the interval $[-C\sqrt{k},C\sqrt{k}]$ for some sufficiently large absolute constant $C > 0$.
    \end{itemize}
\end{fact}
The support of $A$ and $B$ corresponds to the roots of the $k$-th and $(k-1)$-th normalized probabilist's Hermite polynomials.

\section{Hardness Under Lattice Assumptions}
\label{sec:crypto_hardness}

In this section, we will prove a slightly more general version of~\cref{thm:crypto_intro}.
We remark that we will not directly work with~\cref{prob:gap_svp,prob:sivp} but rather with the continuous LWE problem introduced in~\cite{CLWE}.
\begin{theorem}
    \label{thm:main_crypto_hardness}
    Let $0 \leq \delta < 1$.
    Let $k,N \in \N$ such that $k \leq O(\sqrt{N})$.
    Assume there is an algorithm that learns the intersection of $k$ halfspaces in dimension $N$ in time $T = N^{o\Paren{\tfrac {k^{1-\delta}} {\Paren{\log k + \log \log N}^{1-\delta}\cdot\log^\delta N}}}$ up to error better than $\tfrac 1 2 -\tfrac 1 {\Omega(T)}$, then there is an algorithm that solves CLWE in dimension $n$ in time $2^{o\Paren{n^{1-\delta}}}$.
    Furthermore, every halfspace in the hard instance has a margin of $\Omega(\tfrac 1 {\sqrt{N}} \cdot \sqrt{\tfrac {\log \log N + \log k} {k \log N}})$.
\end{theorem}

We refer to the discussion around \cref{thm:crypto_intro} for various instantiations of parameters.

Interestingly, our lower bound almost matches known upper bounds for learning large-margin halfspaces in its dependence on the margin parameter.
In particular, \cite{learning_intersections_margins} show an algorithm to learn intersections of halfspaces with margin $\rho$ in time $O(N \cdot (k / \rho)^{k \log k \cdot \log(1/\rho)})$.
In our lower bound instance $\rho = \Omega(1/ {\sqrt{N \cdot k\cdot \log N}})$ and hence their upper bound becomes roughly (for $k \leq \poly(N)$) $O(N^{k \log (k) \cdot [\log N + \log k]})$ while our lower bound reads $N^{\Omega\Paren{\tfrac k {\log k + \log \log N}}}$.
The most striking difference is the factor of $\log N$, however, note that already for $k = \log^{O(1)} N$ both bounds are basically the same up to a small polynomial factor in the exponent of $k$ and logarithmic terms (in $k$) in the exponent.

We will use the following two facts which are straightforward extensions of facts in \cite{tiegel_agnostic_lattices}.
We will prove them in~\cref{sec:missing_lemmas_crypto}.
\begin{restatable}[Adaptation of Theorem 15 in {\cite{tiegel_agnostic_lattices}}]{fact}{clwetotruncatedclwe}
    \label{fact:clwe_to_truncated_clwe}\RestateRemark
    Let $n, m \in \N$ with $2^{o(n)} = m > n$, and let $\gamma, \beta, \e \in \R_{> 0}$ such that $0 \leq \beta \leq \gamma, \beta = \tfrac 1 {\poly(n)}$.
    Assume that there is no $(T + \poly(n,m))$-time distinguisher between
    \[
    \begin{aligned}
        \clwem{m}{\gamma}{\beta} \quad&\text{and}\quad \Paren{N\Paren{0, \tfrac 1 {2\pi} \cdot I_n} \times U\Paren{[0,1)}}^{\otimes m}
    \end{aligned}
    \]
    with advantage $\e$.
    Let $m' = \tfrac m {\poly(n)}$.
    Then there is no $T$-time distinguisher between
    \[
    \begin{aligned}
        \frac 1 2 \cdot \Paren{\nhclwed{\bm w}{\beta}{\gamma}{0}^{(n)}, +1} + \frac 1 2 \cdot \Paren{\nhclwed{\bm w}{\beta}{\gamma}{\tfrac 1 2}^{(n)}, -1} \quad&\text{and}\quad N\Paren{0, \tfrac 1 {2\pi} \cdot I_n} \times \mathrm{Be}\Paren{\frac 1 2}        
    \end{aligned}
    \]
    with advantage $\e-\negl(n)$ that uses at most $m'$ samples.
\end{restatable}

Further, we will use the following fact about the supports of the mixture of homogeneous CLWE distributions.
Its proof is contained in the proof of Lemma 11 in \cite{tiegel_agnostic_lattices}:
\begin{fact}
    \label{fact:support_intersection_distribution}
    Let $S^{(0)},S^{(1)}$ be the support of $\nhclwed{\bm w}{\beta}{\gamma}{0}^{(n)}$ and $\nhclwed{\bm w}{\beta}{\gamma}{\tfrac 1 2}^{(n)}$ respectively.
    Let $\alpha = \tfrac 1 {10} \cdot \tfrac \gamma {\gamma^2 + \beta^2}$ and for $k \in \N$, let $\mu_k^{(0)} = \tfrac \gamma {\gamma^2 + \beta^2} k,\mu_k^{(1/2)} = \tfrac \gamma {\gamma^2 + \beta^2} (k-\tfrac 1 2)$ then
    \begin{align*}
        S^{(0)} &= \bigcup_{k=-n}^n \Set{x \in \R^n \suchthat \iprod{x, w} \in [\mu_k^{(0)}-\alpha,\mu_k^{(0)}+\alpha]} \,, \\
        S^{(1)} &= \bigcup_{k=-n}^n \Set{x \in \R^n \suchthat \iprod{x, w} \in [\mu_k^{(1/2)}-\alpha,\mu_k^{(1/2)}+\alpha]} \,.
    \end{align*}
    Further, $S^{(0)}$ and $S^{(1)}$ are disjoint and at distance at least $\tfrac 1 5 \cdot \tfrac \gamma {\gamma^2 + \beta^2}$.
\end{fact}

\begin{proof}[Proof of Theorem~\ref{thm:main_crypto_hardness}]
    Let $d, k \in \N$ and $0 \leq \delta < 1$ (it might be instructive to think of $\delta = 0$ first).
    For simplicity assume that $2n+1$ is divisible by $d$ and let $k = \tfrac {2n+1} d$.
    Let $m  \leq  N^{o\Paren{\tfrac{k^{1-\delta}}{\Paren{\log k + \log\log N}^{1-\delta} \log^\delta N}}}$ and $\tau \geq \tfrac 5 {\sqrt{m}}$.
    We will choose $N$ such that $m \leq 2^{o(n^{1-\delta})}$.
    It follows by \cref{fact:clwe_to_truncated_clwe}, that if there is an algorithm that can distinguish between
    \begin{align*}
        D^{(p)} = \frac 1 2 \cdot \Paren{\nhclwed{\bm w}{\beta}{\gamma}{0}^{(n)}, +1} + \frac 1 2 \cdot \Paren{\nhclwed{\bm w}{\beta}{\gamma}{1/2}^{(n)}, -1}\quad\quad&\text{and}\quad\quad D^{(n)} = N\Paren{0,\tfrac 1 {2\pi} I_n} \times \BE{\tfrac 1 2}
    \end{align*}
    with probability at least $2/3$ in time $2^{o(n^{1-\delta})}$ and using at most $2^{o(n^{1-\delta})}$ samples, then there also is an algorithm that solves CLWE with probability at least, say, $0.6$ in time $2^{o(n^{1-\delta})}$ and using at most $2^{o(n^{1-\delta})}$ samples.
    We will show that a learning algorithm would imply the former.

    \paragraph{The reduction}
    Suppose we are given $m$ samples
    $((x_i, y_i))_{i=1}^m \in \R^n \times \Set{-1,+1}$ from either of the two distributions.
    We will later choose parameters such that $m = 2^{o(n^{1-\delta})}$.
    Our reduction does the following:
    Let $N = \sum_{j = 0}^{2d} (n+1)^j = \Theta(n^{2d})$.
    We apply the Veronese mapping to the $x_i$, obtaining $((\tilde{x}_i, y_i))_{i=1}^m \in \R^N \times \Set{-1,+1}$, where $\tilde{x}_i = ((1,x_i)^\alpha)_{\card{\alpha} \leq 2d}$.
    For simplicity, assume that $m$ is even.
    We run our learning algorithm on the first $m/2$ samples to obtain a function $\hat{f} \colon \R^N \rightarrow \Set{+1,-1}$.
    Let
    \[
        \widehat{\err\Paren{f}} = \frac 2 m \sum_{i=m/2}^m \Ind\Paren{\hat{f}(x_i') \neq y_i} \,.
    \]
    If $\Abs{\widehat{\err\Paren{f}} - \frac 1 2} > \frac \tau 2$ we output planted and else we output null.

    First assuming that the learner runs in time $N^{o\Paren{\tfrac{k}{\log k + \log\log N}}}$, notice that the procedure described above runs in the same time -- the reduction only add an overhead of $N^{O(1)}$.
    We claim that this total time is equal to $2^{o(n^{1-\delta})}$.
    This in particular implies that the number of samples $m$ processed by the algorithm is at most $2^{o(n^{1-\delta})}$.
    Indeed, using that $\log N = \Theta(d \cdot \log n) = \Theta(\tfrac n k \cdot \log n)$, we obtain
    \begin{align*}
        N^{o\Paren{\tfrac{k^{1-\delta}}{\Paren{\log k + \log\log N}^{1-\delta} \log^\delta N}}} &= \exp\Paren{o\Paren{\Paren{\frac {k \cdot \log N}{\log k + \log\log N}}^{1-\delta}}} \\
        &= \exp\Paren{o\Paren{\Paren{\frac {n \cdot \log n}{\log k + \log n - \log k + \log\log n}}^{1-\delta}}} \\
        &=2^{o\Paren{n^{1-\delta}}}\,.
    \end{align*}
    To argue that it successfully distinguishes between $D^{(p)}$ and $D^{(n)}$, we proceed in two parts.
    If the input comes from $D^{(n)}$, $y_i \sim \Be{\tfrac 1 2}$ and is independent of $x_i'$, hence $\widehat{\err\Paren{f}}$ will be close to $\tfrac 1 2$.
    If the input comes from $D^{(p)}$, we will show that the samples input to our learning algorithm can be realized as the intersection of $k$ halfspaces -- we will assume this for now in the next paragraph.
    Hence, since we assume access to a weak learner, $\widehat{\err\Paren{f}}$ will be sufficiently smaller than $\tfrac 1 2$. 

    Indeed, under both null and planted the random variables $\Ind\Paren{f(x_i') \neq \tilde{y}_i}$ are \iid Bernoulli with some expectation $p_n,p_p \in [0,1]$ respectively.
    Assume $(x_i,y_i) \sim D^{(n)} = N\Paren{0,\tfrac 1 {2\pi} I_n} \times \Be{\tfrac 1 2}$.
    Since $(\tilde{x_i}, y_i) = (g(x_i), y_i)$ for a deterministic function $g$, it follows that $y_i$ is independent of $\tilde{x_i}$.
    Since clearly, $y_i \sim \Be{\tfrac 1 2}$ it follows that $p_n = \tfrac 1 2$.
    By assumption, $p_p \leq \tfrac 1 2 - \tau$ is the success probability of our learning algorithm.
    It follows by Hoeffding's Inequality \cite{hoeffding} and since $\tau \geq \tfrac 5 {\sqrt{m}}$that in either case (i.e., for $p = p_n$ or $p = p_p$) it holds that
    \[
        \Psymb\Paren{\Abs{\widehat{\err\Paren{f}} - p} \geq \tfrac \tau 3} \leq 2\exp\Paren{-\tfrac m 9 \tau^2} \leq 2\exp\Paren{-2.5} \leq \frac 1 3 \,.
    \]
    Hence, under the null distribution we correctly output null with probability at least $1/3$.
    Similarly, since under the planted distribution with probability at least $1/3$
    \[
        \Abs{\widehat{\err\Paren{f}} - \frac 1 2} \geq \Paren{\frac 1 2 - p_p} - \frac \tau 3 > \frac \tau 2
    \]
    we correctly output planted with the same probability.

    \paragraph{The planted distribution is an intersection of $k$ halfspaces}
    Next, assume
    \[(x_i,y_i) \sim D^{(p)} = \frac 1 2 \cdot \Paren{\nhclwed{\bm w}{\beta}{\gamma}{0}^{(n)}, +1} + \frac 1 2 \cdot \Paren{\nhclwed{\bm w}{\beta}{\gamma}{1/2}^{(n)}, -1}\,.\]
    We argue that $D^{(p)}$ can be realized as an intersection of $k$ degree-$d$ polynomial threshold functions.
    That is, we show that there exists polynomials $p_1, \ldots, p_n \colon \R^n \rightarrow \R$ of degree at most $d$, such that for all $(x,y) \sim D^{(p)}$ it holds that $y = 1$ if and only if $p_j(x) \geq 0$ for all $j = 1,\ldots, k$.
    Note that this directly implies that the transformed samples $(\tilde{x},y)$ we feed to our learning algorithm can be realized as an intersection of $k$ halfspaces.
    In particular, the halfspaces correspond to the linearizations of $p_1, \ldots, p_k$.

    Recall that $w$ is the hidden direction in the planted distribution.
    All polynomials $p_j$ will be of the form $p_j(x) = \tilde{p}_j(\iprod{x,w})$ for one-dimensional polynomials $\tilde{p}_1,\ldots,\tilde{p}_k$.
    On a high level, these will trace out the support of the positive and negative examples.
    Indeed, let $\alpha = \tfrac 1 {10} \cdot \tfrac \gamma {\gamma^2 + \beta^2}$ and for $k = -n,\ldots, n$ let $\mu_k^{(0)} = \tfrac \gamma {\gamma^2 + \beta^2} k,\mu_k^{(1/2)} = \tfrac \gamma {\gamma^2 + \beta^2} (k-\tfrac 1 2)$.
    Define
    \[J^{+}_\ell = [\mu_\ell^{(0)}-\alpha,\mu_\ell^{(0)}+\alpha]\,,\quad\quad\quad\quad J^{-}_\ell = [\mu_\ell^{(1/2)}-\alpha,\mu_\ell^{(1/2)}+\alpha]\,.\]
    Recall that $k = \tfrac {2n+1} d$.
    Let $\tilde{p}_j$ be a degree-$2d$ polynomial that is negative on $J^-_{-n+(j-1)\cdot d}, \ldots, J^-_{-n+j\cdot d -1}$, positive on $J^+_{-n+(j-1)\cdot d}, \ldots, J^+_{-n+ j\cdot d-1}$ and positive starting some distance away from the left and right-most ``negative'' interval.

    Let its root be at the midpoints between consecutive intervals and the left-most root at the same distance to the left-most interval.
    Note that by construction, the following two properties hold (the first property also uses that all $\tilde{p}_j$ are positive after their last root)
    \begin{enumerate}
        \item If $z \in J_\ell^+$ for some $\ell$, then $\tilde{p}_j(z) \geq 0$ for all $j \in [k]$,
        \item If $z \in J_\ell^-$ for some $\ell$, then there exists $j^* \in [k]$ such that $\tilde{p}_{j^*}(z) < 0$.
    \end{enumerate}
    Recall from \cref{fact:support_intersection_distribution} that $y = 1$ implies that $\iprod{x,w} \in \cup_{\ell=-n}^n J_\ell^+$ and $y = -1$ implies that $\iprod{x,w} \in \cup_{\ell=-n}^n J_\ell^-$.
    Hence, the two properties above imply that $y = 1$ if and only if $p_j(x) \geq 0$ for all $j \in [k]$.

    \paragraph{Margin}
    It remains to prove the margin property.
    Note that the definition of margin we use requires us to truncate the points.
    Without loss of generality, assume that we have truncated the $x$-marginal of $D^{(p)}$ to points of norm at most $O(\sqrt{n})$ -- this only introduces a change of $2^{-\Theta(n)}$ in TVD since it is easy to check that the $x$-marginal is sub-Gaussian with variance proxy $O(1)$.
    Let $\tilde{x}$ be an arbitrary point of positive density in the final distribution, then $\norm{\tilde{x}}^2 = \sum_{j=0}^{2d} \norm{(1,x)}^{2j} \leq O(n^{2d})$. Hence, $\norm{\tilde{x}} = O(n^{d}) = O(\sqrt{N})$.

    Let $\tilde{x}^+ = ((x^+,1)^\alpha)_{\card{\alpha} \leq 2d},\tilde{x}^- = ((x^-,1)^\alpha)_{\card{\alpha} \leq 2d}$ be arbitrary points of negative and positive signs.
    Since the vector $\tilde{x}$ contains all monomials of degree at most $2d$, it follows that there exists a vector $\tilde{w} \in \mathbb{S}^{N-1}$ such that
    \[
        \Norm{\tilde{x}^+ - \tilde{x}^-} \geq \Abs{\iprod{\tilde{w},\tilde{x}^+}-\iprod{\tilde{w},\tilde{x}^-}} = \Abs{\iprod{w, x^+ }- \iprod{w, x^-}} \geq \Omega\Paren{\frac 1 {\sqrt{n}}} \,,
    \]
    where in the last inequality we used \cref{fact:support_intersection_distribution}.
    Recall from our earlier calculations that $n = \Theta(\tfrac {k \cdot \log N}{\log k  + \log\log N})$.
    It follows that the distribution has margin at least
    \[
        \Omega\Paren{\frac 1 {\sqrt{N \cdot n}}} = \Omega\Paren{\frac 1 {\sqrt{N}} \cdot \sqrt{\frac {\log k + \log \log N} {k \cdot \log N}}} \,.\qedhere
    \]
\end{proof}

\section{SQ Hardness}
\label{sec:sq_hardness}

In this section, we will prove our SQ lower bound (\cref{thm:sq_intro}).
\begin{theorem}
    \label{thm:main_sq_hardness}
    Let $\beta \in(0,\tfrac 1 2)$ be an absolute constant and $k,N \in \N$ be such that $2 \leq k \leq N^\gamma$ for a sufficiently small absolute constant $\gamma$.
    Every SQ algorithm that uses queries of accuracy $\rho = N^{- \Omega(k)}$ and learns intersections of $k$ halfspaces in dimension $N$ up to error better than $1/2 - 4 \rho$ needs at least $2^{N^{\Omega(1)}}$ queries.
\end{theorem}
To favor clarity of exposition and since in our eyes the ``small $k$'' regime is the most interesting one, we have not tried to optimize constants, i.e., $\gamma$.
We will show the theorem above by constructing a distribution over $(x,y) \in \R^N \times \Set{-1,+1}$ that (a) is an intersection of $k$ halfspaces and (b) the conditional distribution of $x$ given $y = -1$ and $y = +1$ respectively (nearly) matches $k$ moments with the standard Gaussian.

In particular, our hard instance will follow the NGCA framework and will be similar to the construction of~\cite{bubeck2019adversarial} -- in their distribution however, the labels are not without noise.
So we will need to slightly modify it.
That is, the distribution conditioned on $y = +1$ and $y = -1$ will be equal to the standard Gaussian distribution except in one direction (the same direction in both cases), and equal to a distribution that nearly matches $k$ moments with $N(0,1)$ along said direction.
We start by describing the distribution along this direction:
From \cref{fact:moment_matching_hermite} we know that there exists discrete distributions $A,B$ supported on at most $k$ points both matching $2k-1$ moments with $N(0,1)$ and such that all points in the union of their supports are at distance at least $\Omega(1/\sqrt{k})$.
Let $\tilde{A}, \tilde{B}$ be the distributions that are obtained from $A,B$ via the following process -- we only describe it for $A$.
First, let $A'$ be the distribution obtained as follows: Let $\delta > 0$. Draw $X \sim A$ and $Z \sim N(0,1)$ independently.
Output $\sqrt{1-\delta} \cdot X + \sqrt{\delta} \cdot Z$.
Note that $A'$ is a mixture of at most $k$ Gaussians.
Second, truncate each component of $A'$ at distance $\tau$ from its mean.
Later we will choose $\delta = \tfrac 1 {k^2 \log N}$ and $\tau = c \cdot \sqrt{\delta \cdot k  \log N}$ (for a small enough absolute constant $c > 0$).
Our family of hard instances $\cD$ can be described as follows:
\begin{enumerate}
    \item Draw $w \sim \mathbb{S}^{n-1}$ uniformly at random.
    \item Let $D_w^{\tilde{A}}$ be the product distribution that is $\tilde{A}$ along $w$ and a standard Gaussian in the complement (and the same for $D_w^{\tilde{B}}$).
    \item Set $D =\tfrac 1 2 \cdot(D_w^{\tilde{A}}, +1) + \tfrac 1 2 \cdot (D_w^{\tilde{B}}, -1)$.
\end{enumerate}
We will use the notation above throughout the rest of this section.

We will use the following theorem to show that $D$ is in fact hard to learn in the SQ model:
It is an instantiation of results from \cite{DK22} (and a slight refinement of \cite{NT22} already implicit in the first work).
See \cref{sec:missing_lemmas} for full details how this follows from their theorems.
\begin{theorem}
    \label{thm:sq_moment_matching}
    Let $\beta \in(0,\tfrac 1 2)$ be an absolute constant and $\beta' < \beta$.
    Let $k,N \in \N$ be such that $k \leq N^\gamma$ for a sufficiently small absolute constant $\gamma$.
    Let $\tilde{A},\tilde{B}$ be two one-dimensional distributions that match $k$ moments with $N(0,1)$ up to error $N^{-\Omega(k)}$ and such that $\chi^2(\tilde{A},N(0,1)),\chi^2(\tilde{B},N(0,1)) \leq 2^{O(k)} \log N$.
    Let the family of distributions $\cD$ be as above.
    Then any SQ algorithm with accuracy $\rho = N^{-\Omega(k)}$ that learns $\cD$ up to error $\tfrac 1 2 - 4 \rho$ needs at least $2^{N^{\Omega(1)}}$ queries.
\end{theorem}

We can now proceed to prove \cref{thm:main_sq_hardness}:
\begin{proof}[Proof of Theorem~\ref{thm:main_sq_hardness}]
    Let $A', B', \tilde{A},\tilde{B}$ be as above.
    As mentioned before, our proof proceeds in two steps:
    First, we show that $D$ corresponds to an intersection of $k$ degree-2 polynomial threshold functions and second, we will appeal to \cref{thm:sq_moment_matching} to show that $D$ is hard to learn.
    Just as in the proof of~\cref{thm:main_crypto_hardness} this will imply the claim by applying the Veronese mapping.
    Note that the blow-up in the dimension is only quadratic and thus can be absorbed in the $\Omega(\cdot)$- and $O(\cdot)$-notation in our theorem statement.
    We first set parameters, let $c > 0$ be a sufficiently small absolute constant, we set
    \begin{align*}
        \delta = \frac 1 {k^2 \log N} \quad\quad&&\text{and}&&\quad\quad\tau = c \cdot\sqrt{\delta k \log N} = \frac c {\sqrt{k}} \,.
    \end{align*}

    \paragraph{The hard instance is an intersection of $k$ degree-2 PTFs}
    Recall that in $A'$ (resp. $B'$) the mixture components have variance $\delta$ and in $\tilde{A}$ (resp. $\tilde{B}$) we truncate them at distance $\tau$ from their means.
    In particular, let $S_A, S_B \sse \R$ be the collection of intervals of length $2\tau$ around the means of the components of $A'$ and $B'$.
    Since by~\cref{fact:moment_matching_hermite} the means of the components (of both $A'$ and $B'$ together) are at least $\Omega(\tfrac 1 {\sqrt{k}})$ apart, we can choose $c$ in the definition of $\tau$ small enough such that the intervals in $S_A \cup S_B$ are disjoint and at distance $\Omega(\tfrac 1 {\sqrt{k}})$.
    Note that by construction, $S_A$ and $S_B$ contain at most $k$ intervals.

    The proof is analogous to \cref{thm:main_crypto_hardness} with the only difference that we will only use degree-2 polynomials.
    Indeed, by construction, for a sample $(x,y) \sim D$, $y = 1$ if and only if $\iprod{x,w} \in S_A$.
    Further $y = -1$ if and only if $\iprod{x,w} \in S_B$.
    Thus, for every interval $I \sse S_B$, consider the polynomial $p_I \colon \R \rightarrow \R$ that is symmetric around the mid-point of $B$, is negative on $I$, and has its roots at half the distance between the end of $I$ and the next interval in $S_A$.
    Note that $p_I$ is negative on $I$ and positive on all other intervals.
    The final choice of degree-2 PTFs is then $\tilde{p}_I \colon \R^N \rightarrow \R$ such that $\tilde{p}_I(x) = p_I(\iprod{x,w})$.
    By construction, if $\iprod{x,w} \in S_A$, $\tilde{p}_I(x) \geq 0$ for all $I$ and if $\iprod{x,w} \in S_B$ there exists $\tilde{p}_I$ such that $\tilde{p}_I(x) < 0$.
    It follows that $D$ corresponds to the intersection of the $\tilde{p}_I$.

    \paragraph{Set-up for SQ lower bound and $\chi^2$-divergence}
    Note that in order to prove \cref{thm:main_sq_hardness} it is now enough to verify that $\cD$ satisfies the conditions of \cref{thm:sq_moment_matching}.
    Since the conditions on $N$ and $k$ are assumed to be true, it only remains to verify the following
    \begin{enumerate}
        \item $\tilde{A}$ and $\tilde{B}$ match $k$ moments with $N(0,1)$ up to error $N^{-\Omega(k)}$,
        \item $\chi^2({\tilde{A},N(0,1)})$ and $\chi^2({\tilde{B},N(0,1)})$ are at most $2^{O(k)} \log N$.
    \end{enumerate}
    We will verify the properties above only for $\tilde{A}$, $\tilde{B}$ is completely analogous.
    Then~\cref{thm:main_sq_hardness} is implied by~\cref{thm:sq_moment_matching}.

    We start with the $\chi^2$-divergence.
    Let $S_A$ be as in the previous paragraphs.
    Note that $\tilde{A}$ is the distribution $A'$ conditioned on lying in $S_A$.
    In particular, it follows that $p_{\tilde{A}}(x) = 1\Paren{x \in S_A} \cdot \tfrac {p_{A'}(x)}{\Psymb_{X \sim A'}(X \in S_A)}$.
    By standard concentration bounds for the Gaussian distribution, it follows that $\Psymb_{X \sim A'}(X \not\in S_A) \leq \exp(-\Omega(\tfrac {\tau^2} \delta))$ and hence also that $\Psymb_{X \sim A'}(X \in S_A) \geq 1 -\exp(-\Omega(\tfrac {\tau^2} \delta)) \geq \tfrac 1 2$.
    Denote the probability density function of $N(0,1)$ by $G$.
    From \cite[Lemma 4.6]{DKS17}, we know that $\chi^2(A',N(0,1)) \leq 2^{O(k)}/\sqrt{\delta}$.
    It follows that
    \begin{align*}
        \chi^2(\tilde{A},N(0,1)) + 1 &= \int_{-\infty}^\infty \frac{p_{\tilde{A}}(x)^2}{G(x)} \,dx = \frac 1 {\Psymb_{X \sim A'}\Paren{X \in S_A}^2} \cdot \int_{S_A}\frac{p_{A'}(x)^2}{G(x)} \,dx \\
        &\leq 4 \cdot \int_{-\infty}^\infty \frac{p_{A'}(x)^2}{G(x)} \,dx \leq 4 \chi^2\Paren{A', N(0,1)} + 4 \\
        &\leq \frac{2^{O(k)}}{\sqrt{\delta}} \,.
    \end{align*}
    Recalling that $\delta = \tfrac 1 {k^2 \log N}$ we obtain that $\chi^2(\tilde{A},N(0,1))  \leq 2^{O(k)} \log N$.

    \paragraph{Moment matching}
    By \cref{fact:moment_matching_hermite} $A$ matches $2k-1$ moments exactly with $N(0,1)$.
    We claim $A'$ does too: Indeed, for every integer $0 \leq \ell \leq 2k-1$ we have (in the following $X,Z,Z'$ are all independent)
    \begin{align*}
        \E_{X' \sim A'} (X')^\ell &= \E_{X \sim A, Z \sim N(0,1)}\Paren{ \sqrt{1-\delta}\cdot X + \sqrt{\delta} \cdot Z}^\ell \\
        &= \sum_{r = 0}^\ell \binom{\ell}{r} \E_{X \sim A}(1-\delta)^{r/2}\cdot X^r \E_{Z \sim N(0,1)}\delta^{(\ell - r)/2} \cdot Z^{\ell-r} \\
        &= \sum_{r = 0}^\ell \binom{\ell}{r} \E_{Z' \sim N(0,1)}(1-\delta)^{r/2}\cdot (Z')^r \E_{Z \sim N(0,1)}\delta^{(\ell - r)/2} \cdot Z^{\ell-r} \\
        &= \E_{Z' \sim N(0,1), Z \sim N(0,1)}\Paren{ \sqrt{1-\delta}\cdot Z' + \sqrt{\delta} \cdot Z}^\ell \\
        &= \E_{Z \sim N(0,1)} Z^\ell \,.
    \end{align*}

    We next show that the moments of $\tilde{A}$ are close to the moments of $A'$.
    We start with some observations:
    First, note that by construction $\Psymb_{\tilde{A}}(X \not\in S) = 0$.
    Second, let $C' > 0$ be a large enough constant, such that all means are at least $2\tau$ away from the boundary of the interval $[-C'\sqrt{k},C'\sqrt{k}]$.
    Note that the density of $\tilde{A}$ is 0 outside this interval by construction.
    Let $\mu_k$ be the mean of the right-most component, by~\cref{fact:moment_matching_hermite} $\mu_k = O(\sqrt{k})$.
    Choose $C'$ such that $C'\sqrt{k} - \mu_k \geq \tau$.
    Since $\tau = \tfrac c {\sqrt{k}}$ for some constant $c$, such a choice of $C' > 0$ exists.
    Note that
    \begin{align*}
        \Psymb_{X \sim A'}\Paren{\Abs{X} \geq C'\sqrt{k}}&\leq O(k)\cdot \Psymb_{X \sim N(\mu_k,\delta)}\Paren{X \geq C' \sqrt{k}} \\
        &\leq O(k) \cdot \exp\Paren{-\frac{\Paren{\mu_k - C'\sqrt{k}}^2}{2\delta}} = \exp\Paren{-\Omega\Paren{\frac{\tau^2}{\delta}}} \,,
    \end{align*}
    where we used that $\tfrac {\tau^2} \delta = \Omega(k \log N) \gg \log k$.

    Lastly, we note that the total variation distance between $A'$ and $\tilde{A}$ is at most $\exp(-\Omega(\tfrac{\tau^2}\delta))$:
    \begin{align*}
        \Norm{p_{A'} - p_{\tilde{A}}}_1 &= \int_{-\infty}^\infty \Abs{p_{A'}(x) - p_{\tilde{A}}(x)} \, dx \\
        &= \int_S \Paren{\frac 1 {\Psymb_{X \sim A'}(X \in S)} - 1} \cdot p_{A'}(x) \, dx + \int_{S^c} p_{A'}(x) \, dx \\
        &= \int_S \frac {\Psymb_{X \sim A'}(X \not\in S)} {\Psymb_{X \sim A'}(X \in S)} \cdot p_{A'}(x) \, dx + \Psymb_{X \sim A'}(X \not\in S) \leq 3\cdot\Psymb_{X \sim A'}(X \not\in S)  \\
        &= \exp\Paren{-\Omega\Paren{\frac{\tau^2} \delta}} \,.
    \end{align*}
    
    Using the above observations, we start our moment calculations.
    Let $0 \leq \ell \leq k$, then
    \begin{align*}
        &\Abs{\E_{N(0,1)} X^\ell - \E_{\tilde{A}} X^\ell} = \Abs{\E_{A'} X^\ell - \E_{\tilde{A}} X^\ell} = \Abs{\int_{-\infty}^\infty x^\ell (p_{A'}(x) - p_{\tilde{A}}(x)) \, dx} \\
        &\leq \Abs{\int_{C'\sqrt{k}}^{\infty} x^\ell p_{A'}(x) \, dx + \int_{-\infty}^{-C'\sqrt{k}} x^\ell p_{A'}(x) \, dx} + \Abs{\int_{-C'\sqrt{k}}^{C'\sqrt{k}} x^\ell (p_{A'}(x) - p_{\tilde{A}}(x)) \, dx} 
    \end{align*}

    For simplicity, assume that $1.5 k$ is an integer.
    Recall that $A'$ matches $2k-1 \geq 1.5 k$ moments with $N(0,1)$. 
    For the first absolute value, we can deduce using H\"older's Inequality with $q = \tfrac 3 2$ and $p = 3$, that
    \begin{align*}
        \int_{C'\sqrt{k}}^{\infty} x^\ell p_{A'}(x) \, dx + \int_{-\infty}^{-C'\sqrt{k}} x^\ell p_{A'}(x) \, dx &= \E_{A'} X^\ell \Ind\Set{\Abs{X} \geq C' \sqrt{k}} \\
        &\leq \Paren{\E_{A'} X^{1.5\ell}}^{\tfrac 2 3} \Paren{\Psymb_{X \sim A'}\Paren{\Abs{X} \geq C' \sqrt{k}}}^{\tfrac 1 3} \\
        &\leq \Paren{\E_{N(0,1)} X^{1.5k}}^{\tfrac 2 3} \Paren{\Psymb_{X \sim A'}\Paren{\Abs{X} \geq C' \sqrt{k}}}^{\tfrac 1 3} \\
        &\leq (2k)^{k} \cdot \exp\Paren{-\Omega\Paren{\frac{\tau^2}{\delta}}} \,,
    \end{align*}
    where we also used that the $k$-th moment of $N(0,1)$ can be upper bounded as $k^{k/2}$.
    Since $\tfrac{\tau^2}{\delta} = \Omega( k \log N )$ and $k \log (2k) \leq 2\gamma \cdot k \log N$ for a sufficiently small constant $\gamma$, it follows that this integral is at most $\exp\paren{-\Omega\paren{\tfrac{\tau^2}{\delta}}}$.
    
    Using the total variation bound, we can bound the second absolute value:
    \begin{align*}
        \Abs{\int_{-C'\sqrt{k}}^{C'\sqrt{k}} x^\ell (p_{A'}(x) - p_{\tilde{A}}(x)) \, dx} &\leq \Paren{C' \sqrt{k}}^\ell \norm{p_{A'} - p_{\tilde{A}}} \leq (C'\sqrt{k})^k\exp\Paren{- \Omega\Paren{\frac{\tau^2} \delta}} \\
        &= \exp\Paren{-\Omega\Paren{\frac {\tau^2} \delta}} \,.
    \end{align*}
    Combining the two above displays and using that $\tfrac {\tau^2} \delta = \Omega( k \log N)$, we obtain that
    \[
        \Abs{\E_{N(0,1)} X^\ell - \E_{\tilde{A}} X^\ell} \leq \exp\Paren{-\Omega\Paren{\frac {\tau^2} \delta}} \leq N^{-\Omega(k)} \,.\qedhere
    \]
\end{proof}

\section*{Acknowledgements}
We thank Kiril Bangachev, Guy Bresler, and Vinod Vaikuntanathan for helpful discussions.
We also thank the reviewers for insightful feedback and pointing out improvements to the writing.



\printbibliography

\newpage
\appendix


\section{Missing Lemmas}
\label{sec:missing_lemmas}

\subsection{Missing Lemmas for SQ-Hardness}
\label{sec:missing_lemmas_sq}

We will formally argue how \cref{thm:sq_moment_matching} follows from the results in \cite{DK22,NT22}.
We start by restating Lemma 4.3 of \cite{NT22}.
We remark that this proof follows almost verbatim the proof of \cite{DK22} but makes certain things more explicit which will be useful for us.
The distribution $D_v^{A,B,p}$ with $p = \tfrac 1 2$ in their lemma corresponds to our $\cD$.
They denote the dimension by $m$ instead of $N$.
We use our notation in the restatement below.
\begin{lemma}[Lemma 4.3 of \cite{NT22}]
    Let $k \in \N$ and $\nu, \rho, c > 0$.
    Let $A,B$ be probability distributions on $\R$ such that their first $k$ moments agree with the first $k$ moments of $N(0,1)$ up to error at most $\nu$ and such that $\chi^2(A,N(0,1))$ and $\chi^2(B,N(0,1))$ are finite.
    Denote $\alpha \coloneqq \chi^2(A,N(0,1)) + \chi^2(B,N(0,1))$ and assume that $\nu^2 + \alpha \cdot c^k \leq \rho$.
    Then, any SQ algorithm which, given access to samples from $\cD$, outputs a hypothesis $h \colon \R^N \rightarrow \Set{-1,+1}$ such that
    \[
        \err_\cD(h) < \frac 1 2 - 4\sqrt{\rho} \,,
    \]
    must either make queries of accuracy better than $2\sqrt{\rho}$ or make at least $2^{c^2 \cdot \Omega(N)}\cdot (\rho/\alpha)$ queries.
\end{lemma}

The proof of \cref{thm:sq_moment_matching} follows mostly by setting parameters:
\begin{proof}
    By assumption, we have $\nu = N^{-\Omega(k)}$ and $\alpha = 2^{O(k)} \log N$.
    Let $0 < \beta < \tfrac 1 2$ be a small enough absolute constant and $c = N^{-\beta}$ such that
    \[
        \alpha \cdot c^k = 2^{O(k)} \log (N) \cdot N^{-\beta k} \leq \tfrac 1 2\rho \,. 
    \]
    Then,
    \[
        \nu^2 + \alpha \cdot c^k \leq N^{-\beta' k} = \rho\,.
    \]
    Thus, by Lemma 4.3 any SQ algorithm that learns to up to error $\tfrac 1 2 - 4 \tau$ for $\tau = \sqrt{\rho}$ must either make queries of accuracy $2\tau$ or must make at least
    \[
        \exp\Paren{N^{-2\beta} \cdot \Omega(N) - \Omega(k \log N)} \cdot \frac{2^{-O(k)}}{\log N} =  \exp\Paren{\Omega(N^{1-2\beta}) - \Omega(k \log N)}
    \]
    queries.
    Since $k \leq N^\gamma$ for a sufficiently small $\gamma$, the above is at least $\exp\Paren{ \Omega(N^{1-2\beta})} = \exp(N^{\Omega(1)})$.

    Since we assumed that our SQ algorithm can make queries of accuracy $N^{\beta'k} > 2 \tau$, it follows that it needs at least $2^{\Omega(\sqrt{N})}$ queries.
\end{proof}

We remark that we make the assumption that our SQ algorithm can make queries of accuracy $N^{-\Omega(k)}$ for the following reason:
Lemma 4.3 of \cite{NT22} uses a reduction from an associated testing problem to learning, we believe this reduction needs at least one query of this high accuracy to work (the same applies to \cite{DK22}).
Such an assumption is not necessary to show hardness for the associated testing problem -- which we believe still captures the essence of the learning problem.

\begin{lemma}
    \label{lem:small_tvd_advantage}
    Let $n \in \N, \e > 0$ and distributions $D_n^0$ and $D_n^1$ be such that there exists no $T$-time distinguisher with advantage at least $\e$ between $D_n^0$ and $D_n^1$.
    Further, let $D_n^{1'}$ be a third distribution such that $\TVD{D_n^1}{D_n^{1'}} = \negl(n)$.
    Then there exists no $T$-time distingiusher with advantage at least $\e - \negl(n)$ between $D_n^0$ and $D_n^{1'}$.
\end{lemma}
\begin{proof}
    Suppose there exists a distinguisher $\cA$ between $D_n^0$ and $D_n^{1'}$ with advantage at least $\e - \negl(n)$.
    Using this distinguisher to distinguish between $D_n^0$ and $D_n^1$ gives advantage
    \begin{align*}
        \Abs{\Psymb_{x \sim D_n^0} \Paren{\cA(x) = 0} - \Psymb_{x \sim D_n^1} \Paren{\cA(x) = 0}} \geq \Abs{\Psymb_{x \sim D_n^0} \Paren{\cA(x) = 0} - \Psymb_{x \sim D_n^{1'}} \Paren{\cA(x) = 0}} + \negl(n) \geq \e
    \end{align*}
    which is a contradiction.
\end{proof}

\subsection{Missing Lemmas for Hardness Under Lattice Assumptions}
\label{sec:missing_lemmas_crypto}

In this section, we will prove \cref{fact:clwe_to_truncated_clwe} restated below.
\clwetotruncatedclwe*

\begin{proof}
    From \cite[Theorem 15]{tiegel_agnostic_lattices} we know that the conclusion is true for 
    \begin{align*}
        \frac 1 2 \cdot \Paren{\nhclwed{\bm w}{\beta}{\gamma}{0}^{(\infty)}, +1} + \frac 1 2 \cdot \Paren{\nhclwed{\bm w}{\beta}{\gamma}{\tfrac 1 2}^{(\infty)}, -1} \quad&\text{and}\quad N\Paren{0, \tfrac 1 {2\pi} \cdot I_n} \times \mathrm{Be}\Paren{\frac 1 2} \,.
    \end{align*}
    Our lemma follows by noting that the total variation distance between $\frac 1 2 \cdot \Paren{\nhclwed{\bm w}{\beta}{\gamma}{0}^{(\infty)}, +1} + \frac 1 2 \cdot \Paren{\nhclwed{\bm w}{\beta}{\gamma}{\tfrac 1 2}^{(\infty)}, -1}$ and $\frac 1 2 \cdot \Paren{\nhclwed{\bm w}{\beta}{\gamma}{0}^{(n)}, +1} + \frac 1 2 \cdot \Paren{\nhclwed{\bm w}{\beta}{\gamma}{\tfrac 1 2}^{(n)}, -1}$ is at most $2^{\Theta(-n)}$, even when considering their $m$-fold product for $m = 2^{o(n)}$.
    We can then imply~\cref{lem:small_tvd_advantage}.
    By triangle inequality, it is enough to show that $\nhclwed{\bm w}{\beta}{\gamma}{0}^{(\infty)}$ and $\nhclwed{\bm w}{\beta}{\gamma}{0}^{(n)}$ and $\nhclwed{\bm w}{\beta}{\gamma}{\tfrac 1 2}^{(\infty)}$ and $\nhclwed{\bm w}{\beta}{\gamma}{\tfrac 1 2}^{(n)}$ satisfy this.
    Without loss of generality consider $\nhclwed{\bm w}{\beta}{\gamma}{0}^{(\infty)}$ and $\nhclwed{\bm w}{\beta}{\gamma}{0}^{(n)}$.
    Note that we can couple these two distributions as follows: We first draw a sample $X$ from $\nhclwed{\bm w}{\beta}{\gamma}{0}^{(\infty)}$.
    If $X$ comes from the central $2n+1$ components we set $X' = X$, else, we resample $X'$ independently of $\nhclwed{\bm w}{\beta}{\gamma}{0}^{(\infty)}$ until it does.
    We output $(X,X')$.
    The marginals are correct by construction.
    Thus, the TVD is at most the probability that the first draw of $X$ does not come from the central $2n+1$ components.
    This probability is at most
    \begin{align*}
        \frac {\sum_{\abs{\ell} > n} \rho_{\sqrt{\beta^2 + \gamma^2}}(\ell)}{\sum_{\ell = -\infty}^\infty \rho_{\sqrt{\beta^2 + \gamma^2}}(\ell)} &= \frac {\sum_{\abs{\ell} > n} \exp\Paren{-\pi\ell^2/(\beta^2+\gamma^2)}}{\sum_{\ell = -\infty}^\infty \exp\Paren{-\pi\ell^2/(\beta^2+\gamma^2)}} \leq \sum_{\abs{\ell} > n} \exp\Paren{-\pi\ell^2/(\beta^2+\gamma^2)} \\
        &\leq \sum_{\abs{\ell} > n} \exp\Paren{- \frac {2\ell^2} n} \leq \exp\Paren{-\frac n {10}} \,. \qedhere
    \end{align*}
\end{proof}

\end{document}